\documentclass[superscriptaddress, cha, longbibliography, twocolumn]{revtex4}

\usepackage{amsmath} 
\usepackage{amsfonts}
\usepackage{amsthm}
\usepackage{bm}
\usepackage{siunitx}
\usepackage{graphicx}

\usepackage[utf8]{inputenc}
\usepackage[T1]{fontenc}

\newtheorem{theorem}{Theorem}
\newtheorem{lemma}{Lemma}

\begin{document}

\title{Dynamic mode decomposition for Koopman spectral analysis of elementary cellular automata}

\author{Keisuke Taga}
\email{tagaksk@akane.waseda.jp}
\affiliation{Department of Physics, School of Advanced Science and Engineering, Waseda University, Tokyo 169-8555, Japan}

\author{Yuzuru Kato}
\affiliation{Department of Complex and Intelligent Systems, School of Systems Information Science, Future University Hakodate, Hokkaido 041-8655, Japan}

\author{Yoshihiro Yamazaki}
\affiliation{Department of Physics, School of Advanced Science and Engineering, Waseda University, Tokyo 169-8555, Japan}

\author{Yoshinobu Kawahara}
\affiliation{Graduate School of Information Science and Technology, Osaka University, Osaka 565-0871, Japan,  
and Center for Advanced Intelligence Project, RIKEN, Tokyo 103-0027, Japan 
}

\author{Hiroya Nakao}
\affiliation{Department of Systems and Control Engineering, School of Engineering, Tokyo Institute of Technology, Tokyo 152-8552, Japan}
\date{\today}

\begin{abstract}
We apply Dynamic Mode Decomposition (DMD) to Elementary Cellular Automata (ECA).
Three types of DMD methods are considered and the reproducibility of the system dynamics and Koopman eigenvalues from observed time series are investigated.
{While} standard DMD fails to reproduce the system dynamics and Koopman eigenvalues associated with a given periodic orbit in some cases, Hankel DMD with delay-embedded time series improves reproducibility.
{However, Hankel DMD} can still fail to reproduce all the Koopman eigenvalues in specific cases. 
We propose an Extended DMD method for ECA that uses nonlinearly transformed time series with discretized Walsh functions and show that it can completely reproduce the dynamics and Koopman eigenvalues.
Linear-algebraic backgrounds for the reproducibility of the system dynamics and Koopman eigenvalues are also discussed.
\end{abstract}

\maketitle

{\bf
Dynamic mode decomposition (DMD) provides a data-driven approach to extracting spectral properties of the Koopman operator, a linear operator describing the time evolution of observables of a dynamical system, from time-series data.
Though DMD is promising for applications to data-driven modeling of nonlinear dynamical systems, the reproducibility of the dynamics and the Koopman eigenvalues is not always clear for 
spatially-extended systems whose dynamical variables depend both on space and time.
In this paper, we apply DMD to Elementary Cellular Automata (ECA), which is the simplest example of 
dynamical systems with only a finite number of discrete states that describe spatio-temporal patterns, and assess the performance of three types of DMD algorithms.
}

\section{Introduction}
The Koopman operator describes the evolution of {\it observables} (observation functions) of dynamical systems~\cite{koopman1931,vonneumann1932,mezic2005spectral,budivsic2012applied,bollt2018matching,mauroy2020koopman,brunton2022modern}.
By lifting the dynamics to a high-dimensional space of observables, 
the {system's behavior} can be described by a linear Koopman operator, even if the original dynamical system is nonlinear.
Thus, we can apply linear spectral methods to analyze the nonlinear time evolution of the system, {providing} powerful approaches to the analysis and control of complex dynamical systems.
However, for continuous-valued dynamical systems, the Koopman operator acts on an infinite-dimensional function space and is generally difficult to analyze.
In particular, for spatially-extended dynamical systems with space-dependent variables described by partial differential equations, explicit spectral properties of the Koopman operator have been discussed only for solvable systems such as the Burger's equation~\cite{nathan2018applied,page2018koopman,nakao2020spectral} or Korteweg-de Vries equation~\cite{kdv1,kdv2}.

Dynamic Mode Decomposition (DMD) provides a numerical method for extracting spectral properties of the Koopman operator from time-series data of dynamical systems~\cite{schmid2010dynamic, tu2013dynamic, williams2015edmd, kutz2016dynamic, kawahara2016dynamic, arbabi2017hdmd, li2017edmd, mauroy2020koopman, taga2021ecakoopman, brunton2022modern}.
Thus, DMD has attracted much interest as a data-driven approach to Koopman spectral analysis.
{It} has been used {to analyze} various dynamical systems such as fluid flows, power grids, and pedestrian dynamics.
For spatially-extended continuous-valued dynamical systems, DMD typically yields a large number of eigenvalues, some of which approximate the Koopman eigenvalues. 
However, in general, it is difficult to verify the accuracy of the estimated eigenvalues and to {identify} the eigenvalues that are {essential} for the dynamics.

Our aim in this research is to assess the performance of different {DMD} algorithms by applying them to Elementary Cellular Automata (ECA)~\cite{wolfram2002new}.
ECA provide the simplest example of dynamical systems with only a finite number of states (finite-state systems) that describe spatio-temporal patterns.
Despite the simplicity of their evolution rules, ECA can exhibit various dynamics including chaotic ones.
ECA have been analyzed as the simplest theoretical model of computation, and they have also been studied as mathematical models of real-world spatial patterns such as the pigmentation of sea shells~\cite{coombes2009geometry}, congestion patterns of traffic flows~\cite{nishinari1998analytical}, and peeling patterns of adhesive tapes~\cite{ohmori2019comments}.
ECA have 256 evolution rules in total and Wolfram's qualitative classification of the ECA rules into four classes is particularly well known~\cite{wolfram2002new}.
In our previous study~\cite{taga2021ecakoopman}, we analyzed the spectral properties of the Koopman operator for ECA.
For ECA, the Koopman operator is finite-dimensional and can be represented by a finite-dimensional matrix (see Sec. II for details)~\cite{taga2021ecakoopman}.
Therefore, we can make rigorous statements {about} the spectral properties of the Koopman operator~\cite{taga2021ecakoopman}, which can be compared with the results of DMD.
Thus, ECA provide a good testbed for assessing various {DMD} algorithms.

In this paper, we apply three types of DMD algorithms to ECA and assess their performance. We consider standard DMD (Exact DMD), Hankel DMD with delay-embedded time series, and Extended DMD with nonlinear basis functions, and investigate their performance in reproducing the system dynamics and Koopman eigenvalues.
This paper is organized as follows. 
The Koopman operator and DMD for ECA are introduced in Secs.~II and~III, respectively. Section~IV presents the results of standard DMD, Hankel DMD, and Extended DMD for ECA. Section~V gives a summary.

\section{Koopman Operator for ECA}
Each rule of ECA describes a 3-neighborhood cellular automaton with binary values on a one-dimensional lattice~\cite{wolfram1983statistical,wolfram2002new,coombes2009geometry,nishinari1998analytical}. 
As an example, the time evolution of ECA with rule $60$ (ECA 60) is presented in Fig.~\ref{fig1}.
The state of each cell at the next time step is determined by the cell itself and the two neighboring cells. This rule can be represented using 8 
{binary} digits, which are used to label the rule. There are $2^8=256$ possible rules for ECA in total.
Though the evolution rules of ECA are simple, they can exhibit rich dynamics including chaotic ones.
For example, Fig.~\ref{fig2} shows the typical dynamics of 
ECA 8, ECA 73, ECA 60, and ECA 54 on a lattice of $N=100$ cells with periodic boundary conditions. Each figure {displays} the typical dynamics of the 4 classes (Class I-IV) of Wolfram's classification~\cite{wolfram2002new}.

\begin{figure}[htbp]
    \centering
    \includegraphics[width=1\linewidth]{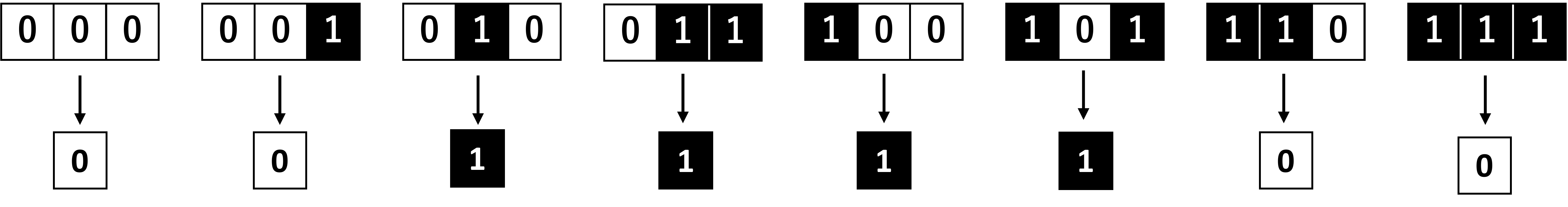}
    \caption{Evolution rule of ECA 60.}
    \label{fig1}
\end{figure}

\begin{figure}[h]
    \centering
    \includegraphics[width=1\linewidth]{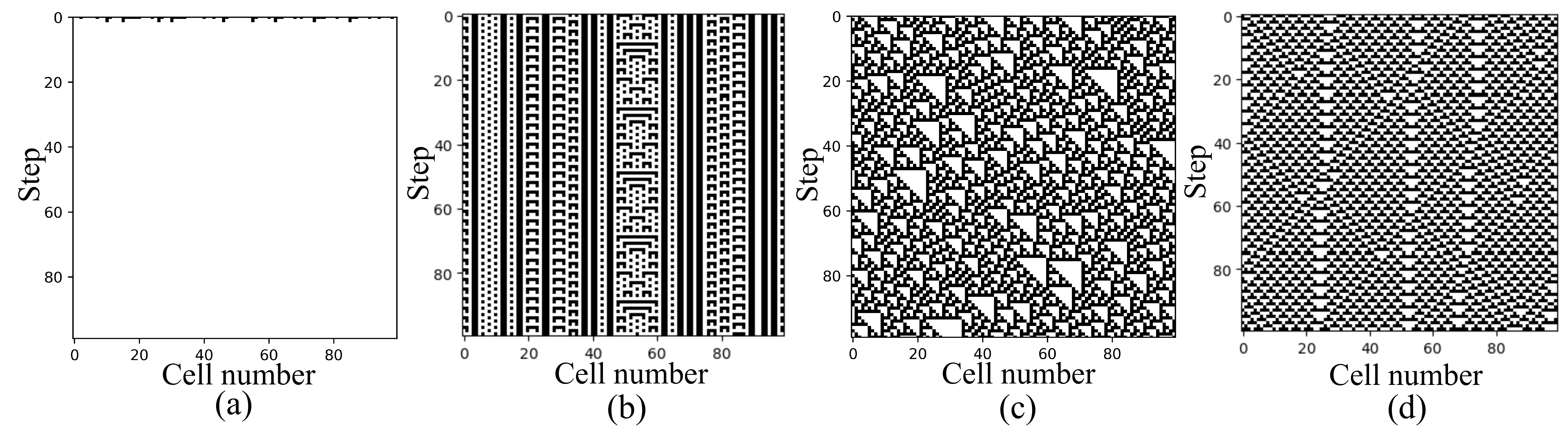}
    \caption{Typical dynamics of ECA. (a) ECA 8 (Class I), (b) ECA 73 (Class II), (c) ECA 60 (Class III), and (d) ECA 54 (Class IV).}
    \label{fig2}
\end{figure}

The Koopman operator is a linear operator that acts on the observables of dynamical systems~\cite{koopman1931, vonneumann1932, mezic2005spectral, budivsic2012applied, bollt2018matching,mauroy2020koopman, brunton2022modern}. 
For a time-discrete dynamical system
\begin{align}
{\bm x}_{n+1}={\bm F}({\bm x}_n),
\end{align}
where ${\bm x}_n \in M$ is a state vector at time step $n$ in the state space $M$ and ${\bm F}:M\to M$ represents the time evolution of the state vector, the Koopman operator $\hat{K}$ is defined by
\begin{align}
\hat{K}[G]({\bm x})=G\circ {\bm F}({\bm x}), 
\end{align}
where $G : M\to {\mathbb C}$ is an observable and $\circ$ represents composition of functions.
This $\hat{K}$ is a linear operator even if ${\bm F}$ is nonlinear, because
\begin{eqnarray}
\begin{split}
\hat K[aG+bH]({\bm x})&
=(aG+bH)\circ {\bm F}({\bm x})\\
&=a\hat K[G]({\bm x})+b\hat K[H]({\bm x})
\end{split}
\end{eqnarray}
holds for general observables $G$ and $H$, and $a, b \in {\mathbb C}$.

The {\it Koopman eigenfunctions} are observables satisfying
\begin{align}
  \hat{K}[\psi]({\bm x}) S{= \psi \circ F(x)} = \lambda \psi({\bm x}),
\end{align}
where $\lambda \in {\mathbb C}$ is a {\it Koopman eigenvalue} and $\psi$ is an associated Koopman eigenfunction.
Assuming that the Koopman eigenfunctions span the function space of observables, we can decompose a general observable $G:M\to{\mathbb C}$ as
\begin{align}
  \label{koopmanmode}
  G({\bm x}) = 
  \sum_{j} g_j \psi_{j}({\bm x}),
\end{align}
where $\lambda_{j}$ is the $j$th eigenvalue, S{$\psi_{j}({\bm x})$} is {the} $j$th Koopman eigenfunction, and $g_j$ is a coefficient called the {\it Koopman mode} ($j=1, 2, ...$). 

The dimensionality of the Koopman operator $\hat{K}$ corresponds to the cardinality $|M|$ of $M$. Therefore, if $M= {\mathbb R}^N$, $\hat{K}$ is an infinite-dimensional operator, and it is generally not easy to perform an explicit analysis of $\hat{K}$ except for solvable systems~\cite{nathan2018applied,page2018koopman,nakao2020spectral}.
In contrast, ECA are finite-state systems because ECA with $N$ cells have only $2^N$ states.
Therefore, we can span the {entire} function space of observables by using $2^N$ indicator functions~\cite{budivsic2012applied,taga2021ecakoopman}, each of which takes a value {of} $1$ only when the system takes one particular state among the $2^N$ states and $0$ otherwise, {providing} a one-hot representation for the whole system states. Using the indicator functions as the basis, the Koopman operator $\hat{K}$ for ECA can be represented as a finite $2^N \times 2^N$ Koopman matrix $K$~\cite{taga2021ecakoopman}.

The properties of the Koopman operator for ECA have been discussed in our previous work~\cite{taga2021ecakoopman}. Some of the important properties are summarized as follows:
\begin{enumerate}
  	\item The Koopman operator $\hat{K}$ has $2^N$ Koopman eigenvalues.
    \item Koopman eigenvalues are either on the complex unit circle or zero.
    \item Koopman eigenvalues on the complex unit circle correspond to the periodic dynamics of the system. If the system has a $T$-periodic solution, $\hat{K}$ has eigenvalues $\exp(2i\pi n/T)$ where $n=1,2,\ldots, T$.
    \item Zero Koopman eigenvalues correspond to relaxation dynamics. If the system has an orbit that reaches a periodic solution in $U$ steps, $\hat{K}$ has a set of zero generalized eigenvalues of every rank from $1$ to $U$.
\end{enumerate}

These properties follow from the fact that the dynamics of any deterministic finite-state systems
reach a periodic solution within S{a finite number of} time steps by the pigeonhole principle~\cite{taga2021ecakoopman} (we regard stationary solutions also as periodic solutions with period one).
The associated eigenfunctions of the Koopman operator represent the 
characteristic dynamical properties of the system.
In particular, the eigenfunctions associated with the eigenvalue $1$ are invariant under the time evolution; namely, they represent conserved quantities of the system~\cite{taga2021ecakoopman}.

The generalized eigenfunctions associated with the relaxation dynamics are the observables that satisfy the following generalized eigenvalue equation:
\begin{align}
  \hat{K}^s [\psi]({\bm x})=0,
\end{align}
  where $s$ is the rank of the generalized eigenfucntion $\psi$.
We can span the {entire} function space of observables with eigenfunctions and generalized eigenfunctions~\cite{taga2021ecakoopman}.
In what follows, to simplify our discussion, we focus only on the asymptotic periodic orbits of ECA and do not consider generalized eigenfunctions
associated with the relaxation dynamics.

\section{Methods of dynamic mode decomposition}
\subsection{Standard DMD}
Dynamic mode decomposition (DMD) is a data-driven method that can approximate spectral properties of the Koopman operator from time-series data observed from dynamical systems~\cite{schmid2010dynamic,tu2013dynamic,williams2015edmd,kutz2016dynamic}.
It is known that, if DMD works appropriately, we can reproduce the dynamics and Koopman eigenvalues of the system from observed time-series data.

Standard DMD seeks the best linear approximation of the evolution of the state vector.
We assign $1$ and $-1$ to black and white, respectively, and use a binary vector ${\bm x} = (x^{(1)}, \ldots, x^{(N)})^{\top} \in \{1, -1\}^{N\times 1}$ to represent the system state, where the $i$th component of ${\bm x}$ corresponds to the state of the $i$th cell of the ECA, and S{$\top$ denotes the transposition}.
We note that $-1$ is used instead of $0$ to represent a white cell in the numerical analysis to avoid the all-zero vector.
In the following analysis, we generally assume ${\bm x} \in {\mathbb R}^{N \times 1}$.

We define a data matrix 
\begin{align}
    X&=({\bm x}_0\ {\bm x}_1\ \ldots\ {\bm x}_{n-1})
    \ {\in {\mathbb R}^{N \times n}}
\end{align}
consisting of $n$ consecutive observed state vectors $\{ {\bm x}_i \}_{i=0, 1, ..., n-1}$, and similarly 
\begin{align}
    Y&=({\bm x}_1\ {\bm x}_2\ \ldots\ {\bm x}_n)
    \ {\in {\mathbb R}^{N \times n}}
\end{align}
consisting of $\{ {\bm x}_i \}_{i=1, ..., n}$, where the time-series data in $Y$ are one step ahead of those in $X$.
The standard {\it Exact DMD} algorithm~\cite{tu2013dynamic} is essentially a spectral analysis of the matrix 
\begin{align}
A=YX^{+} {\in {\mathbb R}^{N \times N}},
    \label{dmdmatrix}
\end{align}
where $X^{+} { \in {\mathbb R}^{n \times N}}$ 
is the Moore-Penrose pseudo-inverse of $X$.
We call this matrix $A$ the {\it DMD matrix}.

We note that the Exact DMD algorithm usually considers dominant components of $YX^{+}$ 
to reduce the computational cost for large systems~\cite{kutz2016dynamic}. 
For ECA with finite states {and} moderate $N$, we can numerically analyze the DMD matrix $A = YX^{+}$ 
without using singular value decomposition to extract the dominant components.
We refer to the above method without dimensionality reduction as the {\it standard DMD} in what follows.

The DMD matrix $A$ minimizes $||AX-Y||$, where the Frobenius norm of a matrix $Z \in {\mathbb R}^{N \times n}$ is defined by $|| Z || = (\sum_{i=1}^N \sum_{j=1}^n |Z_{ij}|^2)^{1/2}$. 
Thus, $A$ is the matrix that gives the best linear approximation of the one-step time evolution of the state vectors included in the data matrix $X$. This linear approximation does not hold for nonlinear systems in general, of course, but it is expected that this gives a reasonable approximation if the dimensionality of the state vector is high enough.
We also note that, even if the time-series data are generated from nonlinear systems, $Y = A X$ can still hold when the matrices $X$ and $Y$ satisfy the linear consistency condition (see below)~\cite{tu2013dynamic}.

If the DMD matrix $A$ can exactly reproduce the dynamics of the observed state vectors, namely, if
\begin{align}
  \label{DMD evolution}
A{\bm x}_n={\bm x}_{n+1}
\end{align}
we can construct the Koopman eigenfunction $\psi$
as
\begin{align}
\label{eigenfunction}
\psi({\bm x}) = \sum_{i=1}^N \phi^{(i)}x^{(i)}=({\bm \phi}, {\bm x}),
\end{align}
where $({\bm a}, {\bm b}) = \sum_{j=1}^N a_j b_j$ is a scalar product of two vectors ${\bm a} = (a_1, ..., a_N)$ and ${\bm b} = (b_1, ..., b_N)$,
$\lambda \in {\mathbb C}$ is an eigenvalue of the transposed matrix $A^{\top}$ of 
$A$,
${\bm \phi} = (\phi^{(1)}, \ldots, \phi^{(N)})^{\top} \in {\mathbb C}^{N \times 1}$ is the associated eigenvector of $A^{\top}$, i.e., $A^{\top} {\bm \phi} = \lambda {\bm \phi}$, 
and $x^{(i)}$ is the $i$th element of ${\bm x}$, i.e., ${\bm x} = (x^{(1)}, ..., x^{(N)})^{\top} \in {\mathbb R}^{N \times 1}$.
Indeed,
\begin{align}
  \hat{K}[ \psi]({\bm x})
  &=
  \psi(A {\bm x}) 
  = \sum_{i=1}^N \phi^{(i)} (A {\bm x})^{(i)} = \sum_{i=1}^N \phi^{(i)} \sum_{j=1}^N A_{ij} x^{(j)}
  \cr
  &
  = \sum_{j=1}^N \left( \sum_{i=1}^N A^{\top}_{ji} \phi^{(i)} \right) x^{(j)} 
  = \sum_{j=1}^N \lambda \phi^{(j)} x^{(j)}
  = \lambda \psi({\bm x}),
  \end{align}
thus $\psi({\bm x})$ is a {\it Koopman eigenfunction} of the system associated with the Koopman eigenvalue $\lambda$.
We also introduce the eigenvector of $A$, called the {\it dynamic mode}, as ${\bm v} = (v^{(1)}, \ldots, v^{(N)})^{\top} \in {\mathbb C}^{N \times 1}$ satisfying $A {\bm v} = \lambda {\bm v}$.

Here, we consider the conditions for the state vectors in the data matrix to satisfy Eq.~(\ref{DMD evolution}).
Let $S$ be a set of state vectors included in the given time series data of ECA used for constructing
the data matrix $X$. 
For the DMD matrix $A$ to reproduce the dynamics of the state vectors, it is enough that $S$ is an invariant set of the ECA dynamics and $Y=AX$ is satisfied.
This is because if every state in $S$ satisfies Eq.~(\ref{DMD evolution}), then S{all states originated} from any initial state ${\bm x}_0$ in $S$ are contained in $S$ and satisfy Eq.~(\ref{DMD evolution}). This discussion is summarized {in} the following lemma.
\begin{lemma}
  \label{timeEvolution}
  If the following conditions are satisfied, the DMD matrix $A$, constructed from data matrices $X$ and $Y$, can reproduce the dynamics of the state vectors in the set of state vectors $S$ included in the given time series.
  \begin{enumerate}
    \item {$S$ is an invariant set of the dynamics.}
    \item {$Y=AX$}
  \end{enumerate}
\end{lemma}

If the DMD matrix $A$ can describe the dynamics, we can construct the Koopman eigenfunctions from the eigenvectors of $A^{\top}$ as in Eq.~(\ref{eigenfunction}) and all the corresponding eigenvalues (except for the irrelevant zero eigenvalues discussed below) of $A$ are Koopman eigenvalues.

We note that, in Ref.~\cite{tu2013dynamic}, Tu {\it et al.} proved that the necessary and sufficient condition for the second assumption, $Y = A X$, is that the matrices $X$ and $Y$ satisfy the `linear consistency' condition, {specifically}, whenever $X{\bm c}=0${,} then $Y{\bm c}=0${,} where ${\bm c} \in {\mathbb R}^N$.
  Here, we simply assume the condition $Y=AX$ and additionally require the existence of an invariant set $S$ of the dynamics. {This is} because we are discussing the reproducibility of the dynamics of finite-state ECA by DMD, {particularly} their periodic orbits. {Hence,} the observed dynamics should be contained in an invariant set $S$ to be reproducible.
In the case that only some elements but not the whole of the state vector
satisfy the linear consistency, we can reproduce the dynamics and associated Koopman eigenvalues only for those elements~\cite{tu2013dynamic}.

In the following discussion, we assume that the given orbit is an asymptotic periodic orbit S{after a sufficiently long transient} in the numerical analysis. In this case, $N$ independent eigenvalues and eigenvectors of $A$ {exist}, consisting of $T'$ eigenvectors associated with non-zero eigenvalues and $N-T'$ eigenvectors associated with zero eigenvalues (see Appendix A). 
Correspondingly, the matrix $A^\top$ also possesses $N$ independent eigenvalues and eigenvectors.
The eigenvectors of $A$ and $A^{\top}$ can be bi-orthonormalized to satisfy 
\begin{align}
({\bm \phi}_{j},  {\bm v}_{k}) = 
\sum_{i=1}^{N} \phi_{j}^{(i)} v_{k}^{(i)} =
\delta_{j, k},
\end{align}
where ${\bm \phi}_{j}$ is the $j$th eigenvector of $A^\top$ with the eigenvalue $\lambda_j\in {\mathbb C}$ and ${\bm v}_{k}$ is the $k$th eigenvector of $A$ with the eigenvalue $\lambda_k\in {\mathbb C}$.
With this normalization, the {matrices}
$\Phi=({\bm \phi}_{1},\ldots,{\bm \phi}_{N})^{\top}$ and $V=({\bm v}_{1},\ldots,{\bm v}_{N})$
satisfy $\Phi V=I$, namely, $\Phi=V^{-1}$ and thus $V \Phi = I$.
Therefore, the following orthogonality holds as well:
\begin{align}
  \sum_{i=1}^{N} \phi_{i}^{(j)} v_{i}^{(k)} =
  \delta_{j, k}.
\label{biorthogonality2}
\end{align}

We can decompose an observable $X^{(i)}({\bm x})$  $(i=1, ..., N$), which gives the $i$th element $x^{(i)}$ of the state vector ${\bm x}${,} as $X^{(i)}({\bm x}) = x^{(i)}$, by
\begin{align}
  \label{dynamicmode}
  X^{(i)}({\bm x}) = \sum_{j=1}^{N} v^{(i)}_{j} \psi_{j}({\bm x}),
\end{align}
namely, the Koopman mode is given by ${\bm v}_{j}$.
Indeed, 
\begin{align}
X^{(i)}({\bm x}) &=
  \sum_{j=1}^{N} v_{j}^{(i)}\psi_{j}({\bm x})
  =\sum_{j=1}^{N} \sum_{k=1}^{N} {v}_{j}^{(i)}\phi^{(k)}_{j}x^{(k)}
  \cr
  &=\sum_{k=1}^{N} \left( \sum_{j=1}^{N} \phi_{j}^{(k)} v_{j}^{(i)} \right) x^{(k)} = \sum_{k=1}^{N}\delta_{i,k}x^{(k)}=x^{(i)},
  \cr
\end{align}
where we used Eq.~(\ref{biorthogonality2}).
We can then express the $i$th element $x_n^{(i)}$ of ${\bm x}_n$ at time step $n$ as
\begin{align}
  x_n^{(i)} 
  &= X^{(i)}({\bm x}_n)
  = \hat{K}^n[ X^{(i)} ]({\bm x}_0) 
= \sum_{j=1}^N v_{j}^{(i)} \hat{K}^n [ \psi_{j} ] ({\bm x}_0)  \cr
  &= \sum_{j=1}^N v_{j}^{(i)} \lambda_j^n \psi_{j}({\bm x}_0) = \sum_{j=1}^N ({\bm \phi}_{j},{\bm x}_0)  \lambda_j^n v_{j}^{(i)},
\end{align}
where ${\bm x}_0$ is the initial state of the system.
In vector form, we obtain
\begin{align}
  \label{vectorform}
  {\bm x}_n 
  &=
\sum_{j=1}^N {\bm v}_{j}\psi_{j}({\bm x}_n) = \sum_{j=1}^N {\bm v}_{j} \hat{K}^n[\psi_{j}]({\bm x}_0)\cr
  &=
   \sum_{j=1}^N ({\bm \phi}_{j}, {\bm x}_0) \lambda_j^n {\bm v}_{j}.
\end{align}
Thus, the periodic dynamics of ECA can be represented as a superposition of linear Koopman modes.

For general state vectors, Eq.~(\ref{DMD evolution}) may not hold 
because the dynamics of ECA is not linear,
and Eqs.~(\ref{eigenfunction}-\ref{vectorform}) may hold only approximately.
However, it can hold for a specific set of state vectors in the given orbit as discussed in Ref.~\cite{tu2013dynamic}.

We note that all zero eigenvalues obtained by DMD with the periodic time-series data are not relevant Koopman eigenvalues. {This is} because 
Eq.~(\ref{eigenfunction}) gives a zero-valued function when ${\bm \phi}_j$ is an eigenvector of $A^{\top}$ with the eigenvalue $\lambda_j = 0$. 
This can be shown as follows.
For the case that the given orbit is $T$-periodic and does not contain relaxation dynamics, the column space of $X$ can be constructed as the linear combination of the eigenvectors of $A$ associated with non-zero eigenvalues (Appendix A), which are orthogonal to ${\bm \phi}_{j}$ with $\lambda_j=0$. 
Therefore, the function constructed by Eq.~(\ref{eigenfunction}) with $\lambda_j=0$ gives a zero-valued function, i.e.,
\begin{align}
  \psi_{j}({\bm x})=({\bm \phi}_{j}, {\bm x})=0
\end{align} 
for ${\bm x} \in S$.
Thus, we can obtain only $N'\ (\leq N)$ sets of Koopman eigenvalues $\lambda_1, ..., \lambda_{N'}$ and eigenfunctions 
$\psi_{1}, ..., \psi_{{N'}}$ out of the $2^N$ sets from $A$.
Namely, Koopman eigenfunctions obtained from $A^{\top}$ are only a part of the Koopman eigenfunctions associated with a given orbit in general.

For ECA, we can always construct an invariant set $S$ in principle, because all states eventually converge to a periodic orbit (including a stationary state) within finite steps.
However, $Y=AX$ is generally not satisfied, e.g., when the period of the given orbit is longer than $N$.
Therefore, the result of standard DMD is only an approximation to the Koopman analysis in general.
This {limits} the reproducibility of the Koopman eigenvalues by standard DMD, as we will see in Sec. IV. 

\subsection{Hankel DMD}

As we will show in Sec. IV, standard DMD cannot always reproduce the system dynamics and Koopman eigenvalues. 
This is {due to} the insufficient dimensionality of the DMD matrix. 
For a finite-state system, we can always include all state vectors 
along a given orbit in the data matrices, in principle.
However, because of the limited dimensionality of the DMD matrix $A$, standard DMD cannot reproduce more than $N$ eigenvalues, while ECA can possess up to $2^N$ Koopman eigenvalues.
Thus, we need to extend standard DMD to obtain more eigenvalues.

Hankel DMD (HDMD) is an extension of standard DMD~\cite{brunton2017havok,arbabi2017hdmd},
which uses enlarged, delay-embedded data matrices (Hankel matrices) defined as
\begin{eqnarray}
  X_H &=& \left(
    \begin{array}{cccc}
      {\bm x}_0 & {\bm x}_1 & \cdots & {\bm x}_{n-1} \\
      {\bm x}_1 & {\bm x}_2 & \cdots & {\bm x}_{n} \\
      \vdots & \vdots & \ddots & \vdots \\
      {\bm x}_p & {\bm x}_{p+1} & \cdots & {\bm x}_{n+p-1}
    \end{array}
  \right),
  \cr \cr
  Y_H &=& \left(
    \begin{array}{cccc}
      {\bm x}_1 & {\bm x}_2 & \cdots & {\bm x}_{n} \\
      {\bm x}_2 & {\bm x}_3 & \cdots & {\bm x}_{n+1} \\
      \vdots & \vdots & \ddots & \vdots \\
      {\bm x}_{p+1} & {\bm x}_{p+2} & \cdots & {\bm x}_{n+p}
    \end{array}
  \right),
 \end{eqnarray}
where $p \geq 1$ is the maximum time delay and $X_H, Y_H \in {\mathbb R}^{(p+1)N\times n}$.
The HDMD matrix $A_H = Y_H X_H^{+} \in {\mathbb R}^{(p+1)N \times (p+1)N}$ is defined {from these matrices and is} used for the DMD analysis instead of the matrix $A$ defined in Eq.~(\ref{dmdmatrix}).
Since up to $(p+1)N$ eigenvalues can be evaluated from $A_H$, the reproducibility of the eigenvalues is improved {compared to} standard DMD.

It has been shown that the HDMD eigenvalues converge to the Koopman eigenvalues if the Hankel matrix contains sufficiently many time-delayed components~\cite{arbabi2017hdmd}. 
For finite-state systems including ECA, we can prove the reproducibility of the dynamics of the state vectors by HDMD with a linear-algebraic argument. 

\begin{lemma}
    \label{Hankel}
  Define $A_H = Y_H X_H^{+}$. Then, the Hankel matrices $X_H$ and $Y_H$ satisfy $Y_H=A_H X_H$ if $X_H$ and $Y_H$ include sufficiently many time-delayed components, specifically, more than the number of independent states in the given orbit.
\end{lemma}
This lemma follows from the {next} lemma. 
\begin{lemma}
  \label{rowspace}
Let $X$ and $Y$ be $n\times m$ matrices. The following statements are equivalent.
\begin{enumerate}
  \item {The row space of $Y$ is included in the row space of $X$.}
  \item {$Y = Y X^{+} X$ holds.}
\end{enumerate}
\end{lemma}

\begin{proof}(Lemma~\ref{rowspace})
  From the property of the pseudo-inverse matrix, ${\bm v}={\bm v}X^{+} X$ holds if and only if the row vector ${\bm v}$ is included in the row space of $X$. 
\end{proof}

Using Lemma~\ref{rowspace}, we can prove Lemma~\ref{Hankel} as follows.

\begin{proof}(Lemma~\ref{Hankel})
Assume that the number of states on the given orbit is $L$. If $p+1 \geq L$, any $n$-step time series of the state vectors on the given orbit is included in the set of row vectors of $X_H$.
Since all row vectors of $Y_H$ are included in the set of $n$-step time series of the state vectors on the given orbit, all row vectors of $Y_H$ are included in the set of row vectors of $X_H$.
Therefore, $Y_H=A_H X_H$ holds from Lemma~\ref{rowspace} if sufficiently many time-delayed components are used for $X_H$ and $Y_H$.
\end{proof}

We note that $Y_H=A_H X_H$ is equivalent to the linear consistency~\cite{tu2013dynamic}. {Thus,} Lemma~\ref{Hankel} proves that we can construct the data matrices $X_H$ and $Y_H$ that satisfy linear consistency by using sufficiently many time-delayed components.

We can always construct the data matrices with an invariant set $S$ by using sufficiently long time-series data observed from the given orbit starting from an initial condition ${\bm x}_0$. 
From Lemma~\ref{timeEvolution} and Lemma~\ref{Hankel}, we obtain the following theorem immediately.
\begin{theorem}
  \label{Hankel DMD}
  HDMD matrices with sufficiently long time-series data and sufficiently many time-delayed components can describe the dynamics of ECA. 
\end{theorem}

Theorem~\ref{Hankel DMD} asserts that HDMD with Hankel matrices constructed from time-delayed components longer than the number of independent state vectors in the given orbit can reproduce the dynamics and thus partially reproduce the spectral properties of the Koopman operator.
We note that Theorem~\ref{Hankel DMD} is valid not only for periodic orbits but also for {orbits that include} relaxation dynamics.
This is a sufficient condition but not a necessary condition{; as} we will show in Sec.~IV~B, in general, we do not need such a large Hankel matrix to reproduce the result of the Koopman analysis.

\subsection{Extended DMD}

As we will show in Sec.~IV~B, if we require all the Koopman eigenvalues associated with a given orbit, 
such as $\lambda_n=\exp(2\pi in/T)$ with every $n=1,\ldots, T$ for the $T$-periodic dynamics, Hankel DMD can still be insufficient.
Therefore, we further consider {the} Extended DMD (EDMD) method for ECA to obtain all the Koopman eigenvalues of the given orbit. 

EDMD is a well-known extension of DMD, which uses a set of nonlinear basis functions for data matrices instead of the state vectors themselves~\cite{williams2015edmd, li2017edmd}.
The basis functions are chosen appropriately so that general observables S{are well approximated by their linear combinations}.
The data matrices $X_E, Y_E \in {\mathbb R}^{m \times n}$ are defined as
\begin{align}
    X_E &= \left(
      \begin{array}{cccc}
        c_1({\bm x}_0) & c_1({\bm x}_1) & \cdots & c_1({\bm x}_{n-1}) \\
        c_2({\bm x}_0) & c_2({\bm x}_1) & \cdots & c_2({\bm x}_{n-1}) \\
        \vdots & \vdots & \ddots & \vdots \\
        c_{m}({\bm x}_0) & c_{m}({\bm x}_1) & \cdots & c_{m}({\bm x}_{n-1})
      \end{array}
    \right),
    \cr
    Y_E &= \left(
      \begin{array}{cccc}
        c_1({\bm x}_1) & c_1({\bm x}_2) & \cdots & c_1({\bm x}_{n}) \\
        c_2({\bm x}_1) & c_2({\bm x}_2) & \cdots & c_2({\bm x}_{n}) \\
        \vdots & \vdots & \ddots & \vdots \\
        c_{m}({\bm x}_1) & c_{m}({\bm x}_2) & \cdots & c_{m}({\bm x}_{n})
      \end{array}
    \right),
   \end{align}
where 
$c_{1}, \ldots, c_{m} : {\mathbb R}^N \to {\mathbb R}$ are generally nonlinear functions.
We then define the EDMD matrix $A_E = Y_E X_E^\dag \in {\mathbb R}^{m \times m}$ and use this for the DMD analysis.

The primary advantage of EDMD is that it can reproduce all Koopman eigenvalues and eigenfunctions of ECA, in principle, while this is not the case for standard DMD or HDMD. By using sufficiently many nonlinear basis functions and {constructing} data matrices containing sufficiently long time-series data, the transpose of the EDMD matrix $A_E$ can always reproduce the matrix representation of the Koopman operator restricted on the set of system states included in the time-series data; the basis with $2^N$ indicator functions~\cite{taga2021ecakoopman} corresponds to the extreme case.
Even if the basis functions are not enough to completely capture the given dynamics of ECA, EDMD is expected to reproduce the Koopman eigenvalues and eigenfunctions more faithfully than standard DMD or HDMD{, provided} the basis functions are appropriately chosen.

\section{Results}

\subsection{Numerical setup}

We now apply three types of DMD methods to ECA.
The observed state vector is 
${\bm x}=(x^{(1)}, \ldots, x^{(N)})^{\top} \in \{1, -1\}^{N\times 1}$.
The dynamics of finite ECA always converge to a periodic orbit after finite steps~\cite{taga2021ecakoopman}. 
We choose a typical initial condition, discard initial relaxation dynamics, and construct the data matrices $X$ and $Y$ using all states on the asymptotic periodic orbit.
{Thus,} we perform DMD using only the time-series data observed from a specific periodic orbit, which is an invariant subset of the state space of the system.

We consider ECA on a lattice of 13 cells with periodic boundary conditions, for which we have calculated all the Koopman eigenvalues of typical periodic orbits from the $2^{13} \times 2^{13}$ Koopman matrix in Ref.~\cite{taga2021ecakoopman}.
We show only the results for some typical orbits of the representative rules that are suitable for demonstrating theoretical results and revealing differences between the three DMD algorithms (see Appendix B for the characteristics of the orbits given by those rules).
It is known that the dynamics can also differ depending on the system size and initial conditions even for the same rule~\cite{nobe2004reversible}, and the following results are only for typical periodic orbits for the illustration purpose. 
We stress, however, that our theoretical arguments regarding the reproducibility of the dynamics and Koopman eigenvalues apply to any rules, orbits, and system sizes.

\subsection{Standard DMD}

We first apply standard DMD to ECA. 
Figures~\ref{fig3},~\ref{fig4}, and~\ref{fig5} show the results for ECA {rules} 4, 184, 3, 73, 15, 54, 30, and 60, respectively. 
S{Here, we chose these rules because they generate typical}
S{orbits suitable for demonstrating the differences in the}
S{performance between the DMD algorithms. Detailed}
S{reasons for choosing these ECA rules are explained in}
S{Appendix B.}
For each rule, the time evolution of the state vector obtained by direct numerical simulations (DNS) of ECA, the time evolution of the state vector predicted by standard DMD, the eigenvalues of the DMD matrix on the complex plane, and DMD modes (except ECA 30 and 60) are shown.

In the numerical analysis, the evolution of the state vector and Koopman eigenvalues are relatively well reproduced by DMD when the period of the system dynamics is less than or equal to $N$, while DMD cannot completely reproduce them when the period is longer than $N$.
Since the dimensionality of the observed state vector ${\bm x}$ is equal to the number of cells $N$, the maximal $\mbox{rank}$ of the DMD matrix $A$ is also $N$. The rank of $A$ is the number of non-zero eigenvalues of $A${, and if it} is smaller than the number of state vectors on the given orbit, DMD fails to reproduce all the Koopman eigenvalues associated with the orbit{, yielding} only {a subset} of them or spurious eigenvalues. 

Figure~\ref{fig3} shows the cases where DMD succeeds in reproducing all the Koopman eigenvalues associated with a given orbit. The dynamics of the state vector are also successfully reproduced by DMD in these cases.
ECA 4 converges to a fixed state (period 1) with Koopman eigenvalue $1$, 
and ECA 184 converges to a periodic orbit of period $T=13$ with Koopman eigenvalues $\lambda_n = \exp(2i\pi n/13)${, where} $n=0,1,...,12$, respectively. For both rules, the Koopman eigenvalues are successfully reproduced by DMD. 

For the rules shown in Fig.~\ref{fig4}, the dynamics of the state vector are also successfully reproduced by DMD, but DMD fails to reproduce some of the Koopman eigenvalues associated with the given orbit. 
For ECA 15, the period is $T=26$, so the rank of $A$ cannot be larger than the period and we cannot obtain all the Koopman eigenvalues associated with the orbit. 
However, we still obtain a part of the Koopman eigenvalues from DMD even though the period $T$ is larger than $N$.

For ECA 3 and 73, even though the period $T$ is less than $13$ in both cases, $\lambda_0=1$ is missing in the case of ECA 3, and $\lambda_1 =\exp(i\pi/3)$ and $\lambda_5=\exp(5i\pi/3)$ are missing in the case of ECA 73.
This is because the observed state vectors from the periodic orbit are not linearly independent and the rank of $A$ is smaller than $T$.
It is interesting to note that, for ECA 73, the system dynamics are divided into 2-periodic dynamics and 3-periodic dynamics as shown in Fig.~\ref{fig4}(b) and the corresponding dynamic modes capture the characteristic spatial structures.

Figure~\ref{fig5} shows the cases {where} DMD fails to reproduce the dynamics. In these cases, we obtain spurious eigenvalues.
For ECA 30 and 60, the periods are $T=260$ and $819$, respectively. Since they are much longer than $N=13$, the dynamics are not reproduced{,} and all the obtained eigenvalues are spurious.
For ECA 54, although the period is $T=4$ and the number of different Koopman eigenvalues is $4$, only one Koopman eigenvalue $\lambda_2=\exp(\pi i)$ is correctly reproduced{,} and some spurious eigenvalues appeared. This is because $Y=AX$ is not satisfied even though the period is less than $N$.
The eigenvalues other than $\lambda_2$ are spurious, but the corresponding dynamic modes appear to show some characteristic spatial structures of the dynamics.

Thus, standard DMD can reproduce some of the Koopman eigenvalues when the dynamics of the observables can be reproduced with the DMD matrix $A$. 
However, even if this condition is satisfied, standard DMD does not always reproduce all the Koopman eigenvalues and, in some cases, it can yield spurious eigenvalues that are not included in the Koopman eigenvalues.
These results are also summarized in Table.~\ref{table:DMD}. 
We note that the time-series data of the system states are taken only from the asymptotic periodic orbit and {do} not include the relaxation dynamics. Therefore, the zero eigenvalues obtained by DMD are not the true Koopman eigenvalues and are irrelevant{,} as discussed in Sec.~III.

\begin{table}[h]
  \caption{Standard DMD for ECA.}
  \label{table:DMD}
  \centering
  \begin{tabular}{lcccc}
    \hline
       ECA & period $T$  & reproducibility of the Koopman analysis  \\
    \hline
    4 & N>T & completely reproduced \\
    184 & N=T & completely reproduced  \\
    3, 73 & N>T & partly reproduced  \\
    15 & N<T & partly reproduced \\
    54 & N>T & spurious eigenvalues are obtained \\
    30, 60  & N<T & spurious eigenvalues are obtained \\
    \hline
  \end{tabular}
\end{table}

\begin{figure}[h]
  \begin{center}
  (a) ECA 4\\
  \includegraphics[width=0.7\linewidth]{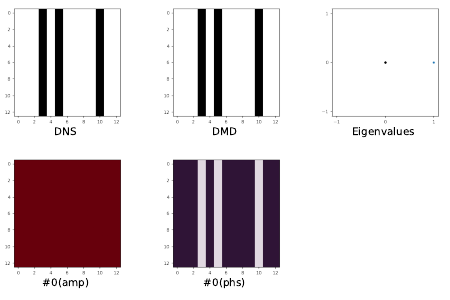}\\
  (b) ECA 184\\
  \includegraphics[width=0.7\linewidth]{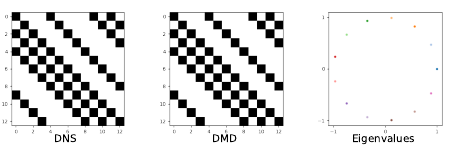}\\
  \includegraphics[width=0.7\linewidth]{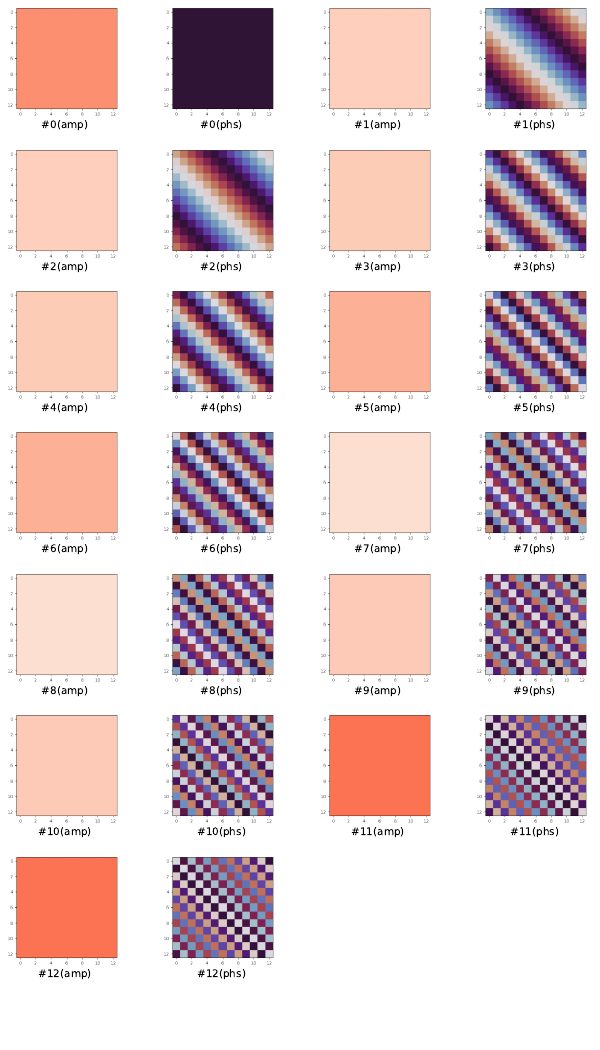}\\
  \includegraphics[width=0.7\linewidth]{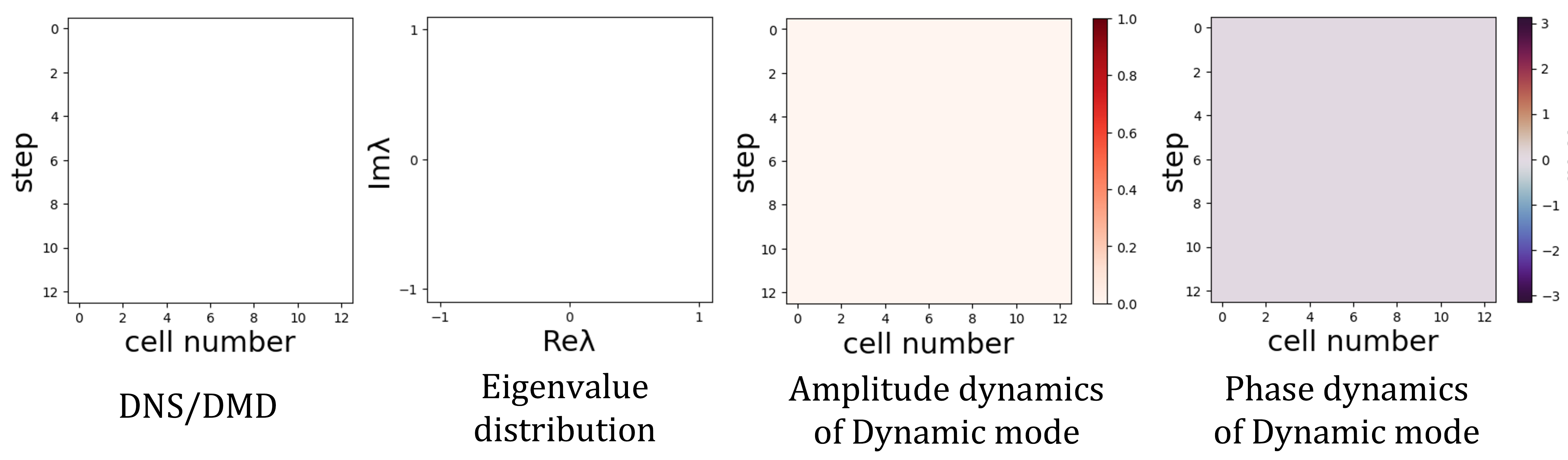}
  \end{center}
  \caption{Standard DMD of ECA on a lattice of $N=13$ cells with periodic boundary conditions. (a) ECA 4 (period $T=1$), (b) ECA 184 ($T=13$).
  The first three figures show the evolution of the state vector obtained by DNS of ECA, the evolution of the state vector predicted by DMD,
  and the eigenvalues of the DMD matrix on the complex plane.
  The rest of the figures show the amplitude and the phase components of all dynamic modes with non-zero eigenvalues. 
  The axes of the individual figures are indicated in the bottom panel. 
  }
  \label{fig3}
  \end{figure}

  \begin{figure}[h]
    \begin{center}
      (a) ECA 3\\
      \includegraphics[width=0.7\linewidth]{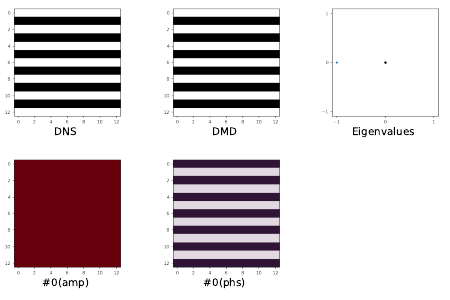}\\
      (b) ECA 73\\
      \includegraphics[width=0.7\linewidth]{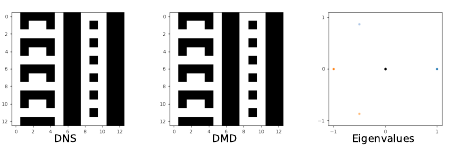}\\
      \includegraphics[width=0.7\linewidth]{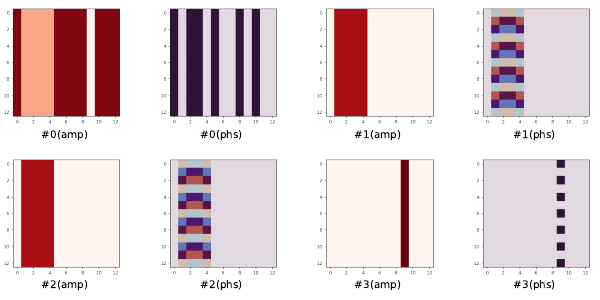}\\
      (c) ECA 15\\
      \includegraphics[width=0.7\linewidth]{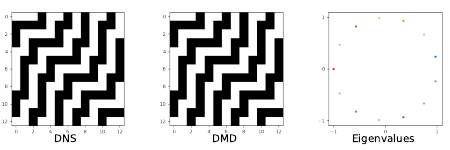}\\
      \includegraphics[width=0.7\linewidth]{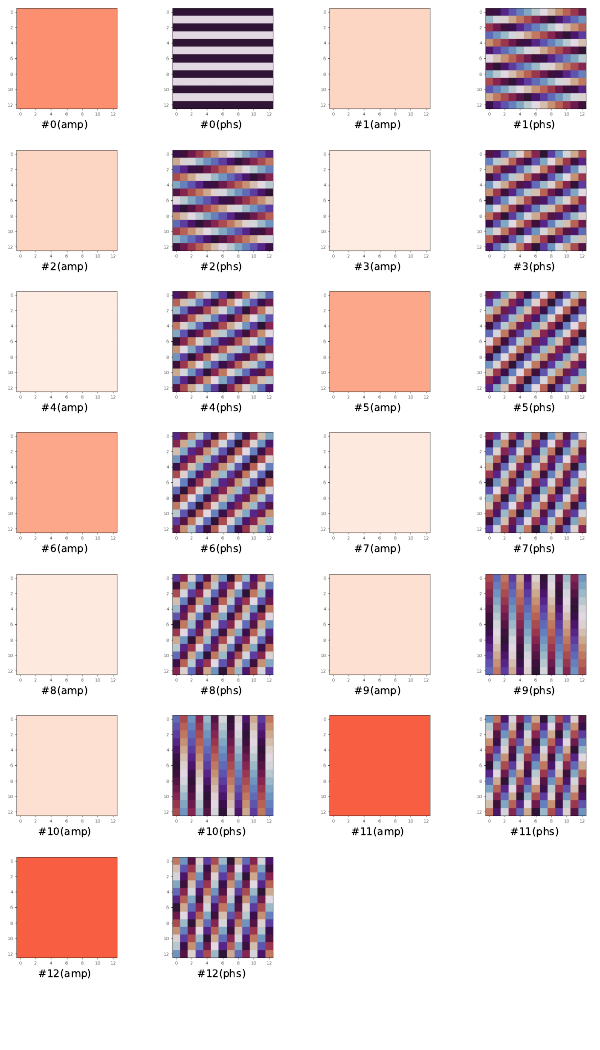}
    \end{center}
    \caption{Standard DMD of ECA on a lattice of $N=13$ cells with periodic boundary conditions. 
    The dynamics are reproduced in these cases, but not all Koopman eigenvalues are reproduced. 
    (a) ECA 3 ($T=2$), (b) ECA 73 ($T=6$), (c) ECA 15 ($T=26$). Each figure is plotted in the same manner as in Fig.~\ref{fig3}.}

    \label{fig4}
  \end{figure}

  \begin{figure}[h]
    \begin{center}
      (a) ECA 30\\
      \includegraphics[width=0.7\linewidth]{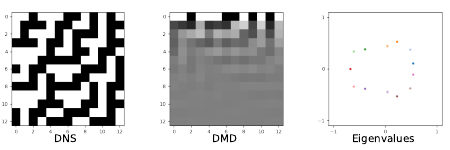}\\
      (b) ECA 60\\
      \includegraphics[width=0.7\linewidth]{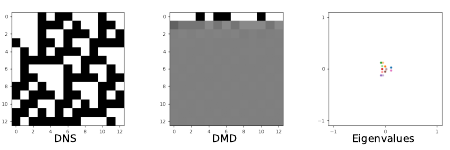}\\
      (c) ECA 54\\
      \includegraphics[width=0.7\linewidth]{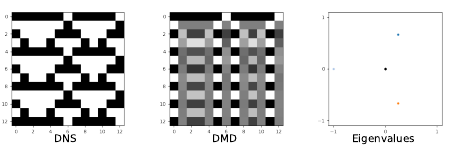}\\
      \includegraphics[width=0.7\linewidth]{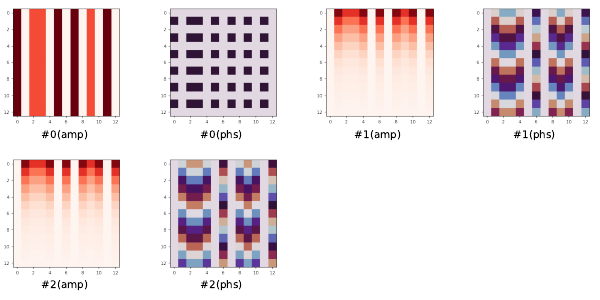}
    \end{center}
    \caption{Standard DMD of ECA on a lattice of $N=13$ cells with periodic boundary conditions. 
    In these cases, the dynamics are not reproduced and spurious Koopman eigenvalues are obtained.
    (a) ECA 30 ($T=260$), ({b}) ECA 60 ($T=819$), and ({c}) ECA 54 ($T=4$). DMD modes are shown only for ECA 54. Each figure is plotted in the same manner as in Fig.{~\ref{fig3}}. 
    }

    \label{fig5}
  \end{figure}

\subsection{Hankel DMD}

Next, we apply Hankel DMD to ECA.
The observed state vector 
${\bm x}=(x^{(1)}, \ldots, x^{(N)})^{\top} \in \{1, -1\}^{N\times 1}$
is the same as in the previous subsection.
The Hankel matrices are constructed using the maximal time delay of $p$ steps. 
It is expected that HDMD with a sufficiently large amount of time-delayed data can reproduce the dynamics and Koopman eigenvalues better than standard DMD.
From Theorem~\ref{Hankel DMD}, 
it is enough to take $p$ as the period of the asymptotic periodic orbit to reproduce the dynamics of the system using HDMD.

Figure~\ref{fig6} shows the evolution of the state vector obtained by DNS of ECA,
the evolution of the state vector predicted by HDMD,
and the eigenvalues of the HDMD matrix on the complex plane for the rules shown in Figs.~\ref{fig4} and~\ref{fig5}.
Note that standard DMD fails to reproduce all the Koopman eigenvalues for these cases.

In contrast to standard DMD, HDMD can predict the evolution of the state vector for all rules shown in Figs.~\ref{fig4} and~\ref{fig5}.
In particular, the evolution of the state vector for ECA 30, 60, and 54, whose dynamics cannot be described by standard DMD, {is} completely reproduced by HDMD.
Also, the Koopman eigenvalues for ECA 30, 60, and 54, which cannot be reproduced by standard DMD, are well reproduced, and no spurious eigenvalues other than $0$ are obtained.
Note here that the zero eigenvalues obtained by HDMD are irrelevant to the Koopman eigenvalues{,} as in the previous case of standard DMD.
For ECA 30 and 54, we can obtain all the Koopman eigenvalues of the given periodic orbit.

For ECA 60, HDMD reproduces most of the eigenvalues, but the number of obtained DMD eigenvalues is $807$, which is smaller than the total number, $819$, of the Koopman eigenvalues associated with the given periodic orbit.
We note that the Koopman eigenvalues $\lambda_n=\exp(2\pi i n/13),n=1,2,...,12$ are missing in this case. This is because the state space 
can be spanned by $807$ state vectors; namely, the state vectors obtained from the periodic orbit are not linearly independent, and if we try to construct corresponding eigenvectors as in Appendix A, we obtain zero vectors.
  These results are summarized in Table.~\ref{table:HDMD}.
\begin{table}[h]
  \caption{Hankel DMD for ECA.}
  \label{table:HDMD}
  \centering
  \begin{tabular}{lcccc}
    \hline
       ECA & period $T$  & reproducibility of the Koopman analysis  \\
    \hline
    4 & N>T & completely reproduced \\
    184 & N=T & completely reproduced  \\
    3, 73 & N>T & partly reproduced  \\
    15, 60 & N<T & partly reproduced \\
    54 & N>T & completely reproduced \\
    30  & N<T & completely reproduced \\
    \hline
  \end{tabular}
\end{table}

As we have shown in Theorem~\ref{Hankel DMD}, if we introduce delay-embedded time series longer than the number of states 
on the given periodic orbit, then $A_H$ can reproduce the dynamics.
Thus, for ECA 54, we need $p=4$ to describe the $4$-periodic dynamics with $A_H$.
However, this condition is sufficient but not necessary.
For example, ECA 30 in Fig.~\ref{fig6}(d) shows the $260$-periodic dynamics, namely, it is sufficient to reproduce the dynamics with $p=260$ from Theorem~\ref{Hankel DMD}. However, we need just $p=20$ to reproduce the dynamics in this case. 

In contrast to the above cases, for the rules shown in Figs.~\ref{fig6}(a), (b), and (c), we obtain the same sets of S{Koopman eigenvalues as those obtained} by standard DMD except for $0$.
This is because if $A_H$ with a delay $p$ can describe the dynamics, we obtain the same nonzero eigenvalues from $A_H$ with a longer delay $p' (>p)$.
Namely, no more Koopman eigenvalues of the periodic orbit can be obtained even if we use HDMD for those cases.  
Therefore, for ECA 60, we cannot obtain more Koopman eigenvalues even if we use a larger Hankel matrix. 

In summary, HDMD can reproduce the dynamics of ECA and improves the reproducibility of the eigenvalues.
In particular, all the eigenvalues other than $0$ obtained by HDMD with sufficiently large Hankel data matrices are Koopman eigenvalues. However, if we require a complete set of the Koopman eigenvalues and Koopman modes for the given periodic orbit, HDMD is still insufficient.

\begin{figure}[h]
  \centering
  (a) ECA 3\\
  \includegraphics[width=0.7\linewidth]{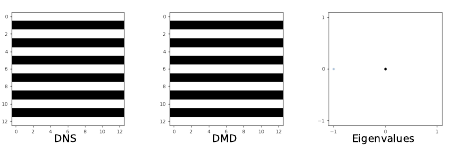}\\
  (b) ECA 73\\
  \includegraphics[width=0.7\linewidth]{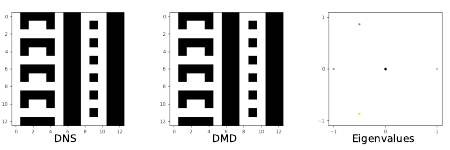}\\
  (c) ECA 15\\
  \includegraphics[width=0.7\linewidth]{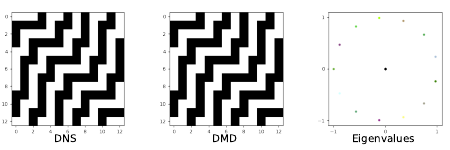}\\
  (d) ECA 30\\
  \includegraphics[width=0.7\linewidth]{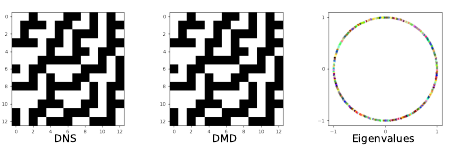}\\
  (e) ECA 60\\
  \includegraphics[width=0.7\linewidth]{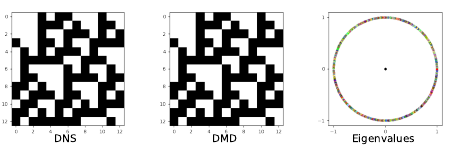}\\
  (f) ECA 54\\
  \includegraphics[width=0.7\linewidth]{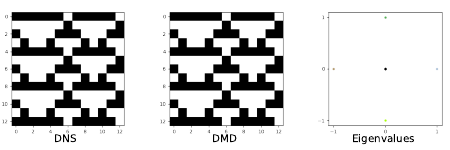}\\
  \includegraphics[width=0.7\linewidth]{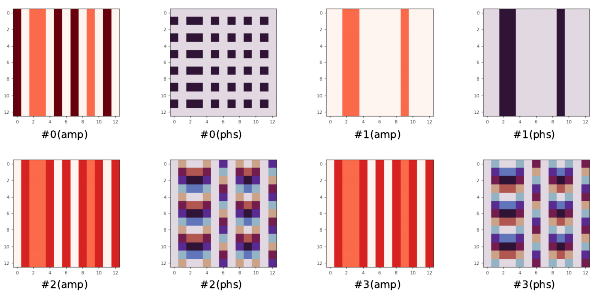}
  \caption{HDMD of ECA on a lattice with $N=13$ cells with periodic boundary conditions. (a) ECA 3 (delay $p+1=2$), (b) ECA 73 ($p+1=6$), (c) ECA 15 ($p+1=26$), (d) ECA 30 ($p+1=20$), (e) ECA 60 ($p+1=63$), and (f) ECA 54 ($p+1=4$). 
  For each rule, the evolution of the state vector obtained by DNS of ECA, the evolution of the state vector predicted by HDMD, 
  and the eigenvalues of the HDMD matrix on the complex plane are plotted.
  For ECA 54, the dynamic modes are also plotted (the dynamics modes corresponding to delayed components are not shown).
  }
  \label{fig6}
\end{figure}

\subsection{Extended DMD}
Finally, we apply EDMD to ECA. To use EDMD, we need to choose nonlinear observables appropriate for ECA. 
In our previous study~\cite{taga2021ecakoopman}, we used indicator functions~\cite{budivsic2012applied} for Koopman analysis, which map the system state to a number $1,\ldots, 2^N$ without considering the spatial configuration of $N$-cell ECA. However, since the evolution rules of ECA are spatially local, function systems with spatially localized properties are more efficient.
In this study, we propose the following nonlinear observables:
\begin{align}
c_{q}({\bm x})=\Pi_{j=1} (x^{(j)})^{\alpha_{j}^{(q)}},
\label{nonlinearobservables}
\end{align}
where $x^{(j)} \in \{1, -1\}$ 
is the $j$th element of the observed state vector ${\bm x} \in \{1, -1\}^N$, $\alpha_{j}^{(q)} \in\{0,1\}$, and the index $q$ is given by 
\begin{align}
    q = \sum_{j=1}^N 2^{j-1} \alpha_{j}^{(q)} + 1.
\end{align}
We show an example in Fig.~\ref{edmdVector}.
\begin{figure}[h]
  \centering
  \includegraphics[width=1\linewidth]{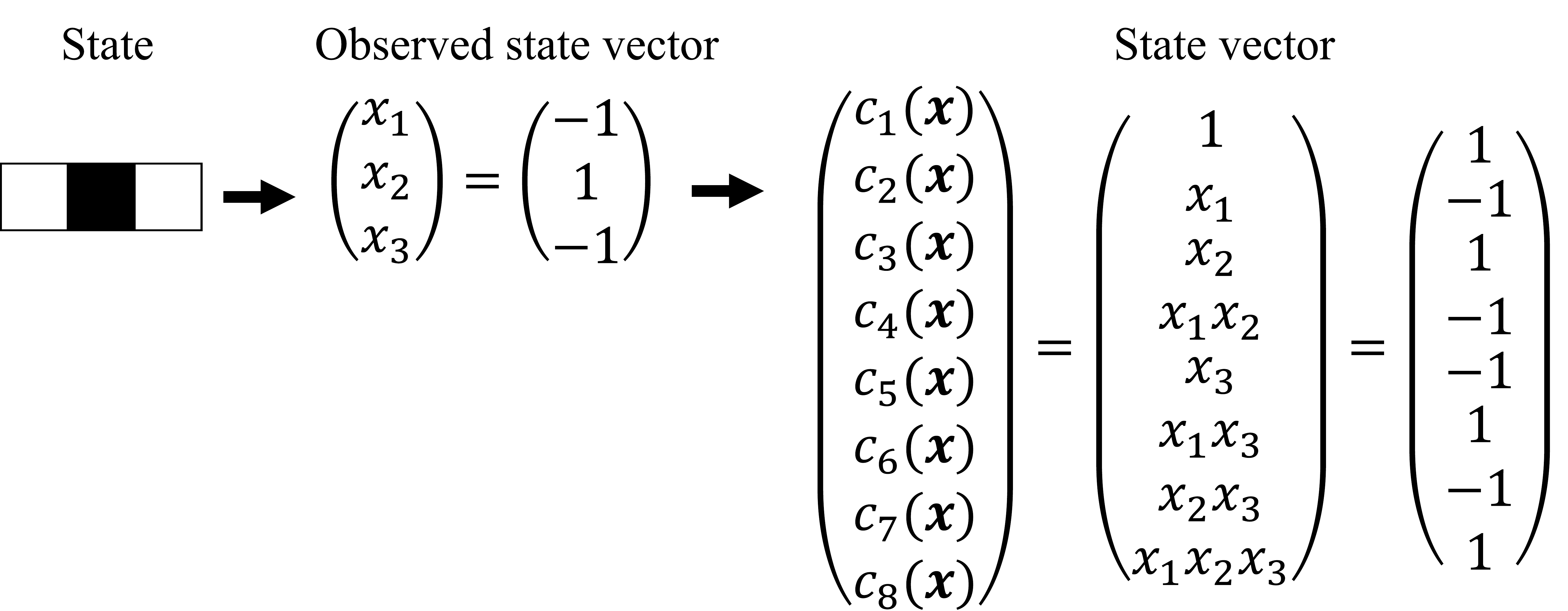}
  \caption{Example of the state vector for EDMD with the 3-cell system.}
  \label{edmdVector}
\end{figure}
If we can take $\{ \alpha_{j}^{(q)} \}$ arbitrarily, the above definition gives $2^N$ different observables $c_1({\bm x}), ..., c_{2^N}({\bm x})$, which form an orthonormal basis of the function space defined on $M$.
These nonlinear observation functions are related to the Walsh function system $\{w_n\}_{n=1,2,\ldots}$~\cite{paley1932orthogonal}, which is a complete orthonormal system defined on $[0,1)$.

The Walsh function system can be constructed using the Rademacher function system~\cite{paley1932orthogonal}. The Rademacher function system $\{r_n\}_{n=1,2,\ldots}$ is an orthonormal but not complete function system defined on $[0,1)$, {and} is defined as follows:
\begin{eqnarray}
r_0(t)&=&1,\\
r_1(t)&=&\left\{\begin{split}-1,\quad t\in[0,0.5),\\1,\quad x\in[0.5,1),\end{split}\right.\\
r_n(t)&=&r_1(t/2^{n-1}) \quad (n>1).
\end{eqnarray}
Here, the domain of $r_1(t)$ is extended periodically for $t>1$. 
The Walsh function system is defined as
\begin{align}
  w_q({\bm x})=\Pi_{j=1}r_{j}^{\alpha_{j}^{(q)}},
\end{align}
where $\alpha^{(q)}_j$ is the same as defined above.
By using the Rademacher functions, we can represent $x^{(k)}$
as
$x^{(k)}=r_k(\frac{u({\bm x})}{2^N})$,
where we introduced the index of each state ${\bm x}$ as 
\begin{align}
 u({\bm x}) = \sum_{j=1}^N 2^{j-1} \frac{x^{(j)}+1}{2}.
\end{align}
Using this index, the nonlinear functions $c_k({\bm x})$ can be represented with the Walsh functions as $c_k({\bm x})=w_k(\frac{u(\bm x)}{2^N})$.

If we can use all of the $2^N$ observables, we can lift the $N$-dimensional state of ECA to {a} $2^N$-dimensional function space and reproduce all the Koopman eigenvalues associated with a given orbit, in principle.
In practice, this quickly becomes impractical as $N$ increases. {However, if we consider only a given $T$-periodic orbit, the number of independent states on the orbit is $T$.
Therefore, we need only $T$ nonlinear functions at least to reproduce the dynamics and Koopman eigenvalues associated with the given orbit. }

Thus, to reduce the computational cost, we also consider the case that only a subset of the nonlinear functions can be used. 
This gives an upper limit to the maximal number of eigenvalues that can be reproduced by EDMD{, namely, to reproduce all the Koopman eigenvalues associated with a given orbit, the number of functions in the subset should be at least $T$.}
%
In what follows, we use a set of nonlinear functions $\left\{ c_{q}({\bm x})=\Pi_{j=1} (x^{(j)})^{\alpha_{j}^{(q)}}: \sum_{j=1}^N \alpha_j^{(q)} \leq m'\right\}$ for EDMD, where $m'\in \{1,\ldots, N\}$.
The total number of the functions is given by $\sum_{j=0}^{m'}{}_{N} C_j$, where ${}_N C_j$ represents a binomial coefficient.
We note that the eigenvalues $0$ obtained by EDMD are irrelevant to the Koopman eigenvalues as discussed in Sec. III for standard DMD.

Figure~\ref{fig8} shows the evolution of the state vectors obtained by DNS of ECA with ECA 3, 73, 15, 30, 60, and 54, the evolution of the state vector predicted by EDMD with the above set of functions,
and the eigenvalues of the EDMD matrix on the complex plane.
The dynamic modes are shown only for ECA 54.
The reproduced state of the $k$th cell can be obtained from $c_{2^{k-1}}({\bm x})$ for 
$k=1,\ldots,N$, because the following equation holds:
\begin{align}
  x^{(k)}=r_k\left(\frac{u({\bm x})}{2^N}\right)=w_{2^{k-1}}\left(\frac{u({\bm x})}{2^N}\right)=c_{2^{k-1}}({\bm x}).
\end{align}
The dynamics and Koopman eigenvalues of ECA are now completely reproduced by EDMD. 
For ECA 3, 73, 15, 60, for which both standard DMD and HDMD fail, all the Koopman eigenvalues of the given periodic orbit are reproduced by EDMD.
In particular, for ECA 3, we can reproduce the Koopman eigenvalues not by considering the nonlinear function of ${\bm x}$ but by simply adding a constant function to the data matrix of standard DMD.

However, for ECA 60, a large number of nonlinear functions, $\sum_{j=1}^N \alpha_j^{(q)} \leq 4$, are necessary.
Figure~\ref{fig9} 
shows the evolution of the state vector predicted by EDMD (top row) and the eigenvalues of the EDMD matrix on the complex plane (bottom row) for ECA 60. In this case, the number of nonlinear functions is $\sum_{j=0}^{4}{}_{N} C_j = 1093$ 
for the case of $N=13$.
This is larger than the period $T=819$, namely, it is enough to reproduce the dynamics and the Koopman eigenvalues associated with the given periodic orbit.

The nonlinear observables used in this subsection can depend on the cell states irrespective of their spatial locations; {hence,} they use `global' information of the system dynamics. In Appendix C, we examine a different choice of 'local' nonlinear observables that depend only on the neighboring cell states and briefly discuss the difference between the two types of observables.

\begin{figure}[h]
  \centering
  (a) ECA 3\\
  \includegraphics[width=0.7\linewidth]{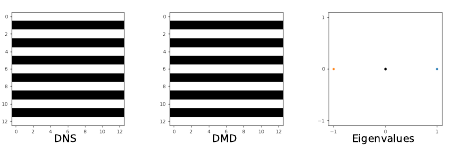}\\
  (b) ECA 73\\
  \includegraphics[width=0.7\linewidth]{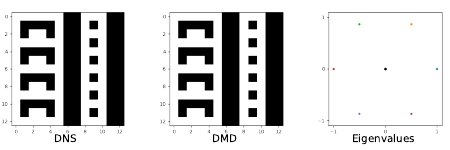}\\
  (c) ECA 15\\
  \includegraphics[width=0.7\linewidth]{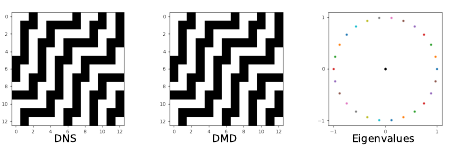}\\
  (d) ECA 30\\
  \includegraphics[width=0.7\linewidth]{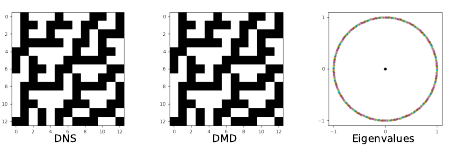}\\
  (e) ECA 60\\
  \includegraphics[width=0.7\linewidth]{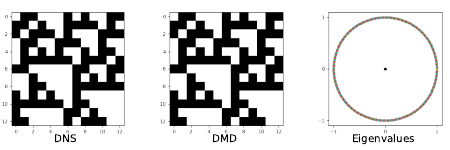}\\
  (f) ECA 54\\
  \includegraphics[width=0.7\linewidth]{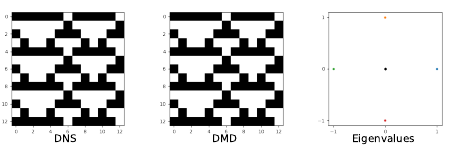}\\
  \includegraphics[width=0.7\linewidth]{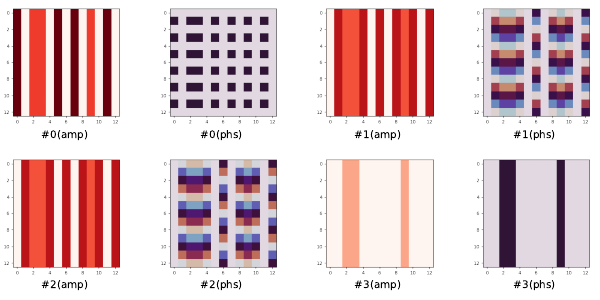}
  \caption{EDMD of ECA on a lattice of $N=13$ cells with periodic boundary conditions. (a) ECA 3 ($m'=1, \#c'=14$), (b) ECA 73 ($m'=2, \#c'=92$), (c) ECA 15 ($m'=2, \#c'=92$), (d) ECA 30 ($m'=3, \#c'=378$), (e) ECA 60 ($m'=4, \#c'=1093$), and (f) ECA 54 ($m'=2, \#c'=92$). For each rule, the figures show the evolution of the state vector obtained by DNS of ECA, the evolution of the state vector predicted by EDMD,
  and the eigenvalues of the EDMD matrix on the complex plane.
  For ECA 54, the dynamic modes are also shown,
  where only the elements $c_k({\bm x}), k=2^{n-1}, n=1,\ldots, N$,
  which correspond to the actual values of the cells,
  are plotted. $\#c'$ is the number of nonlinear functions for each $m'$.
  }
  \label{fig8}
\end{figure}

\begin{figure}[h]
  \centering
  \includegraphics[width=\linewidth]{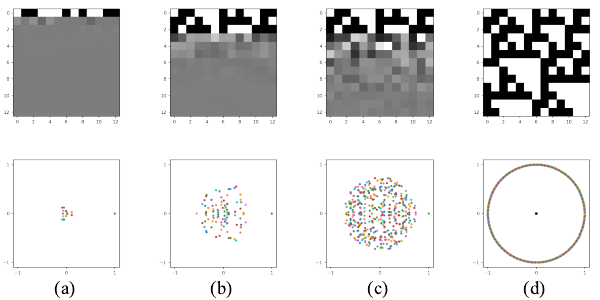}
  \caption{EDMD of ECA 60 on a lattice of $N=13$ cells with periodic boundary conditions. (a) $m'=1$, (b) $m'= 2$, (c) $m'=3$, and (d) $m'= 4$. Figures in the top row show the evolution of the state vector predicted by EDMD, where the same initial conditions are used. Figures in the bottom row show the eigenvalues of the EDMD matrix on the complex plane.}
  \label{fig9}
\end{figure}

\section{Discussion}
We have performed a DMD analysis of ECA. We considered three types of DMD methods, i.e., standard DMD, Hankel DMD, and Extended DMD. 

Standard DMD could not reproduce the results of the Koopman operator analysis when the period of the dynamics is longer than the number of cells. In particular,
it could yield spurious eigenvalues that are nonexistent in the true Koopman eigenvalues in some cases. 
This indicates that care must be taken in applying DMD to extract the Koopman eigenvalues of real-world systems from time-series data.

{Hankel DMD, employing} a delay-embedded enlarged data matrix{,} improved the reproducibility of dynamics and  eigenvalues. {However, it} still failed to reproduce some of the eigenvalues of the given periodic orbit {under certain conditions}. 

{In contrast, Extended DMD} using discrete Walsh-like nonlinear functions could completely reproduce the dynamics and all Koopman eigenvalues for the observed periodic orbit when a sufficiently large number of observation functions are used.
Though we {theoretically} need all $2^N$ observables to completely span the {entire} function space of $N$-cell ECA, {in practice, using only a subset of} nonlinear functions with $\sum_{j=1}^N \alpha_j^{(q)} \leq m'${yielded} reasonably good results for many rules.

The reproducibility of dynamics and Koopman eigenvalues by each DMD algorithm reflects the linear algebraic properties of the DMD matrix calculated from the data matrices.
{While our focus in this study was} only on the asymptotic periodic dynamics of ECA for clarity{, }reproducibility of DMD analysis for finite-state systems including the transient dynamics should be {investigated} in more detail in future works.

In this paper, we have demonstrated that ECA can be used as a good testbed for investigating DMD algorithms and discussed the theoretical conditions for the reproducibility of the dynamics and eigenvalues.
It would be an interesting future topic to perform a more exhaustive numerical analysis of ECA with various system sizes and initial conditions
to investigate whether DMD can be used to gain new insights into the properties of ECA.
 
\section*{Acknowledgments}

We acknowledge JSPS KAKENHI JP22K11919, JP22H00516, JP22K14274, and JST CREST JP-MJCR1913 for financial support.

\section*{Conflicts of interest}

The authors have no conflicts to disclose.

\section*{Data availability}

The data that supports the findings of this study are available within the article.

\appendix
\section{Eigenvectors of $A$}
In the case that the data matrices $X, Y$ are constructed {solely from} an asymptotic $T$-periodic orbit $({\bm x}_0\to\ldots\to{\bm x}_{T-1}\to {\bm x}_T={\bm x}_0)$ and Eq.~(\ref{DMD evolution}) is satisfied, we can construct all the eigenvectors of $A$ as follows. 
First, all vectors in the column space of $X$, {denoted as} $\mbox{Col}(X)$, are on the $T$-periodic orbit and satisfy $A^T{\bm x}={\bm x}$ because ${\bm x}_{n+T}={\bm x}_n$.
Thus, if we construct a vector
\begin{align}
  {\bm v}=\sum_{n=0}^{T-1} \lambda^{-n}{\bm x}_n,
  \label{eqa1}
\end{align}
it satisfies
\begin{align}
A {\bm v}
 &= \sum_{n=0}^{T-1} \lambda^{-n} A {\bm x}_n
 = \sum_{n=0}^{T-1} \lambda^{-n} {\bm x}_{n+1}
= \lambda \sum_{n'=1}^{T} \lambda^{-n'} {\bm x}_{n'}
\end{align}
and, if $\lambda^T = 1$,
\begin{align}
A {\bm v} = \lambda \sum_{n'=0}^{T-1} \lambda^{-n'} {\bm x}_{n'}
 = \lambda {\bm v},
 \label{eqa2}
\end{align}
where we used ${\bm x}_T = {\bm x}_0$. 
Therefore, if ${\bm v} \neq 0$, ${\bm v}$ is an eigenvector of $A$ with the eigenvalue $\lambda$ satisfying $\lambda^T = 1$.
Since $A^T {\bm v} = \lambda^T {\bm v}$ and $A^T {\bm v} = \sum_{n=0}^{T-1} \lambda^{-n} A^T {\bm x}_n = \sum_{n=0}^{T-1} \lambda^{-n} {\bm x}_n = {\bm v}$, $\lambda$ should satisfy $\lambda^T = 1$ when ${\bm v} \neq 0$.
Therefore, the associated eigenvalue $\lambda_j$ of $A$ can be either of $\lambda_j=\exp(2j\pi i/T)$ $(j=0, \ldots, T-1)$.
We note that, even if $\lambda_j \neq 0$, ${\bm v}_j$ defined in Eq.~(\ref{eqa1}) can be a zero-valued vector when ${\bm x}_0, \ldots, {\bm x}_{T-1}$ are not linearly independent. If so, ${\bm v}_j$ is not an eigenvector of $A$ even if Eq.~(\ref{eqa2}) is satisfied.

The transformation from the set $\{{\bm x}_n\}_{n=0,\ldots,T-1}$ to the set $\{{\bm v}_{j}\}_{j=0,\ldots,T-1}$ can be expressed by using a matrix $U$ as
\begin{align}
  ({\bm v}_{0},\ldots,{\bm v}_{{T-1}})^{\top} = U({\bm x}_0,\ldots, {\bm x}_{T-1})^{\top},
\end{align}
where the $(j,k)$-element of $U$ is given by
\begin{align}
  U_{jk}=\exp \left( \frac{2(j-1)(k-1)\pi i}{T} \right).
\end{align}
The inverse of $U$ is given as $\frac{1}{T}U^{\dagger}$, where $U^{\dagger}$ is a Hermitian conjugate of $U$. Indeed,
\begin{align}
  \begin{split}
  &\left(U \frac{1}{T}U^{\dag}\right)_{jl}\\
  &=\sum_{k=1}^T\exp \left( \frac{2(j-1)(k-1)\pi i}{T} \right) \frac{1}{T}\exp \left( - \frac{2(k-1)(l-1)\pi i}{T} \right)\\
  &=\frac{1}{T}\sum_{k=1}^T \exp \left( \frac{2(j-l)(k-1)\pi i}{T} \right) = \delta_{jl}.
  \end{split}
\end{align}
Therefore, the state vectors in $X$ can be constructed as the linear combination of $\{{\bm v}_{j}\}_{j=0,\ldots,T-1}$ as
\begin{align}
  ({\bm x}_0,\ldots,{\bm x}_{T-1}) = \frac{1}{T}U^{\dagger}({\bm v}_{0},\ldots,{\bm v}_{{T-1}})^{\top}.
\end{align}
Therefore, $\mbox{Col}(X)$ corresponds to the space spanned by the set of vectors $\{{\bm v}_{j}\}_{j=0,\ldots,T-1}$.

Here, as stated previously, some of $\{{\bm v}_{j}\}_{j=0,\ldots,T-1}$ can be zero vectors.
That is, if the dimension $T'$ of the vector space spanned by $({\bm x}_0,\ldots,{\bm x}_{T-1})$ is smaller than $T$, 
namely, if ${\bm x}_0, ..., {\bm x}_{T-1}$ are not linearly independent, 
then $T-T'$ vectors among the $T$ vectors $\{{\bm v}_{j}\}_{j=0,\ldots,T-1}$ in Eq.~(\ref{eqa1}) are zero-valued vectors.
This is because the vectors $\{ {\bm v}_{j} \}$
are eigenvectors of $A$ corresponding to different eigenvalues if they are non-zero vectors, so they should be linearly independent of each other if they are non-zero vectors. 
Therefore, among $T$ vectors, only $T'$ vectors are non-zero and $T-T'$ vectors must be zero-valued vectors. Thus, we have obtained $T'$ independent eigenvectors that span $\mbox{Col}(X)$.

Next, regarding the vectors in the orthogonal complement $\mbox{Col}(X)^\perp$ of $\mbox{Col}(X)$, they are the eigenvectors of $A$ associated with the zero eigenvalues. 
$\mbox{Col}(X)^\perp$ is the kernel of $X^{+}$ from the property of the pseudo-inverse matrix~\cite{brata2012pseudoinverse},
and the kernel of $X^{+}$ is contained in the kernel of $A=YX^{+}$, namely, $A{\bm v}=0$ holds for $^\forall{\bm v} \in \mbox{Col}(X)^\perp$. As the dimension of $\mbox{Col}(X)^\perp$ is $N-T'$, we can take $N-T'$ independent eigenvectors associated with the zero eigenvalue that span $\mbox{Col}(X)^\perp$.

Summarizing, we have constructed all the eigenvectors of $A$ that can span the {entire} $N$-dimensional vector space, namely, the $T'$ independent eigenvectors associated with non-zero eigenvalues that span $\mbox{Col}(X)$ and $N-T'$ independent eigenvectors associated with the zero eigenvalues that span $\mbox{Col}(X)^\perp$.

\section{ECA rules used as examples}

We consider periodic orbits of 8 ECA rules as examples to illustrate the differences in the performance of DMD.
Their characteristics are briefly summarized as follows.

\begin{itemize}

\item ECA 4 and 184 are chosen as examples where standard DMD is sufficient for Koopman analysis.
\begin{itemize}
\item	ECA 4 gives an example of static dynamics.
\item	ECA 184 gives an example of short-periodic dynamics.
\end{itemize}

\item ECA 3, 73, and 15 are chosen as examples where standard DMD can reproduce the dynamics but cannot reproduce some of the Koopman eigenvalues associated with the given orbits, despite the short periodicity of the orbit. To reproduce all Koopman eigenvalues, EDMD should be used for these rules.
\begin{itemize}
\item	ECA 3 gives an example where all cells exhibit the same periodic dynamics. 
\item	ECA 73 gives an example where multiple groups of cells coexist that exhibit different periodic dynamics.
\item	ECA 15 gives an example where the length of the period is longer than the number of cells but standard DMD succeeds in reproducing the dynamics.
\end{itemize}
 
\item ECA 54, 30, and 60 are chosen as examples where standard DMD gives spurious eigenvalues, namely, standard DMD does not reproduce the dynamics and HDMD or EDMD is required.
\begin{itemize}
\item	ECA 54 gives an example of short-period dynamics where HDMD requires a delay embedding of the same length as the periodic orbit, that is, it requires a delay embedding of the maximum length that satisfies the conditions in Lemma 2.
\item	ECA 30 gives an example of chaotic long-period dynamics where HDMD can reproduce all Koopman eigenvalues for the given orbit and requires a delay embedding of a length smaller than the period.
\item	ECA 60 gives an example of chaotic long-period dynamics. HDMD cannot reproduce all the Koopman eigenvalues for the given orbit, and EDMD is required to reproduce them.
\end{itemize}
\end{itemize}

\section{EDMD with other nonlinear basis functions}

As discussed in Sec.~IV D, we can always reproduce all the Koopman eigenvalues of a given orbit by using a complete set of basis functions as observables for EDMD. In Sec.~IV D, we considered a set of observation functions that do not form a complete basis and investigated {the extent to which} they can reproduce the Koopman eigenvalues, where the observation functions are 'global' in the sense that they can depend on arbitrary cell states of the system. Here, we consider a different choice of the basis functions {that} are 'local' in the sense that they depend only on the neighboring cell states and examine how efficiently (i.e., without using large data matrices) we can reproduce the results of the Koopman analysis.

Here, instead of the set of functions
$\left\{ c_{q}({\bm x})=\Pi_{j=1} (x^{(j)})^{\alpha_{j}^{(q)}}: \sum_{j=1}^N \alpha_j^{(q)} \leq m'\right\}$
in Sec.~IV~D, 
we consider another set of functions 
\begin{multline}
  \left\{ c_{q}({\bm x})=\Pi_{j=1} (x^{(j)})^{\alpha_{j}^{(q)}} : \right.\\
  \left.\alpha_k^{(q)}=0\ \mbox{if}\; ^\exists i,\ ^\forall k,\ \  i-k\ \mbox{mod}\ N\geq m'' \right\}
\end{multline}
where $i,k,m''\in \{1, \ldots, N\}$ and $\alpha_k^{(q)}\in\{0,1\}${. These functions} take into account the cell states only within the $m''$ nearest neighbors of each cell{, providing} a set of 'local' Walsh functions that depend only on the states of the cells up to $m''$ nearest neighbors.

Figure~\ref{fig10} shows the results for the dynamics and eigenvalues obtained by DNS and reproduced by EDMD for ECA 73, 15, 30, 60, and 54{. These results} can be compared with the results in Sec.~IV D (Fig.~\ref{fig8}).
For ECA 15 and ECA 54, the dynamics can be reproduced by using the cell states {within only} two nearest neighbors, which is more efficient than the case with the original basis functions used in Sec.~IV. 
{However,} for ECA 30 and ECA 60, similar numbers of basis functions {are needed} as in the case with the original basis functions {to reproduce the dynamics}.

The difference between {these} two cases is as follows. The dynamics of ECA 15 and 54 preserve local structures of the patterns, while the dynamics of ECA 30 and ECA 60 cause chaotic changes {in} the overall structures of the patterns. 
This observation suggests that for rules preserving local structures, the set of 'local' functions that take into account only the states of neighboring cells is more efficient for EDMD.

For ECA 73, a large neighborhood size $m''=7$ {is needed} to reproduce the Koopman eigenvalues of the given orbit. This is because the pattern is composed of  $2$-periodic and $3$-periodic regions of sizes $6$, but the data matrix for $m''=6$ does not satisfy the condition 2 of Lemma 1, and therefore a set of functions up to $m''=7$ is necessary to satisfy the condition.
Thus, while the dynamics can be easily reproduced by DMD also for ECA 73, to reproduce the eigenvalues that require information {on} the entire dynamics of the system, it is necessary to use a set of functions that takes into account global information from distant cells.

The above results suggest that a set of basis functions that are efficient for EDMD can be chosen by considering the dynamical structure of the system.
\begin{figure}[h]
  \centering
  (a) ECA 73\\
  \includegraphics[width=0.7\linewidth]{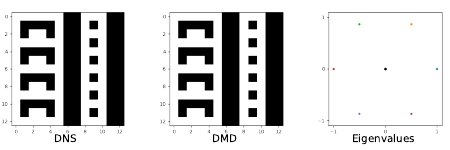}\\
  (b) ECA 15\\
  \includegraphics[width=0.7\linewidth]{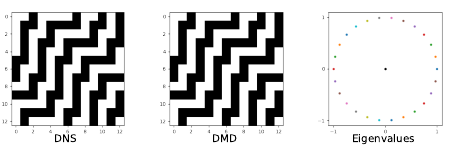}\\
  (c) ECA 30\\
  \includegraphics[width=0.7\linewidth]{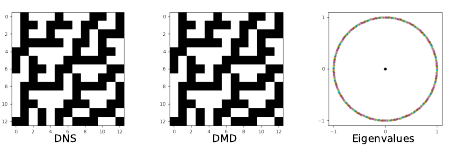}\\
  (d) ECA 60\\
  \includegraphics[width=0.7\linewidth]{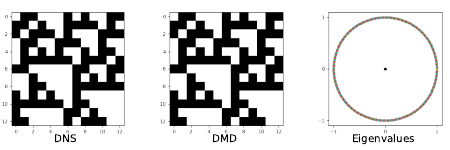}\\
  (e) ECA 54\\
  \includegraphics[width=0.7\linewidth]{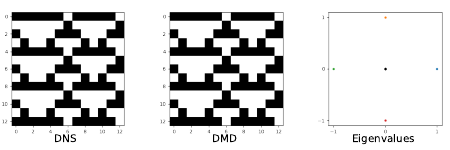}\\
  \caption{EDMD of ECA with a different set of observation functions on a lattice of $N=13$ cells with periodic boundary conditions. 
  (a) ECA 73 ($m''=7, \#c''={833}$), (b) ECA 15 ($m''=2, \#c''=27$), (c) ECA 30 ($m''=6, \#c''=417$), (d) ECA 60 ($m''=8, \#c''={1639}$), and (e) ECA 54 ($m''=2, \#c''=27$). For each rule, the figures show the evolution of the state vector obtained by DNS of ECA, the evolution of the state vector predicted by EDMD,
  and the eigenvalues of the EDMD matrix on the complex plane.
  {Only} the elements $c_k({\bm x}), k=2^{n-1}, n=1,\ldots, N$,
  which correspond to the actual values of the cells, are plotted. The symbol $\#c''$ represents the minimal number of nonlinear functions for each $m''$ required to reproduce the dynamics and eigenvalues.}
  \label{fig10}
\end{figure}

\bibliographystyle{apsrev4-1}
\bibliography{refs}

\end{document}